\documentclass[envcountsame,envcountsect]{llncs}
\usepackage{amssymb,amsmath}
\usepackage{latexsym}
\usepackage{enumerate}
\usepackage[vmargin=1in,hmargin=1in]{geometry}
\usepackage{graphicx}
\usepackage{mathrsfs}
\usepackage{subfigure}
\usepackage{relsize,xspace}
\usepackage[british]{babel}

\usepackage{times}

\def\squishlist{\setlength{\itemsep}{0pt}\setlength{\parsep}{0pt}%
  \setlength{\topsep}{0pt}\setlength{\partopsep}{0pt}\setlength{\parskip}{0pt}}

\let\realbfseries=\bfseries
\def\bfseries{\realbfseries\boldmath}

{\makeatletter
 \gdef\xxxmark{%
   \expandafter\ifx\csname @mpargs\endcsname\relax 
     \expandafter\ifx\csname @captype\endcsname\relax 
       \marginpar{xxx}
     \else
       xxx 
     \fi
   \else
     xxx 
   \fi}
 \gdef\xxx{\@ifnextchar[\xxx@lab\xxx@nolab}
 \long\gdef\xxx@lab[#1]#2{{\bf [\xxxmark #2 ---{\sc #1}]}}
 \long\gdef\xxx@nolab#1{{\bf [\xxxmark #1]}}
}

\newcommand{\set}[1]{\{#1\}}
\newcommand{\st}{\mathrel :}
\newcommand{\Oof}{\mathcal{O}} 
\newcommand{\tw}{\operatorname{tw}}

\newcommand{\floor}[1]{\left\lfloor#1\right\rfloor}
\renewcommand{\bar}[1]{\overline{#1}}

\newcommand{\isom}{\cong}

\newcommand{\minor}{\preccurlyeq}
\newcommand{\dminor}{\minor^d}
\newcommand{\bminor}{\minor^b}
\newcommand{\tminor}{\minor^t}
\newcommand{\undir}[1]{#1^{u}}
\newcommand{\undirG}{\undir{G}}

\newcommand{\bidir}[1]{#1^{bd}}
\newcommand{\bidirG}{\bidir{G}}
\newcommand{\innode}{\operatorname{in}}
\newcommand{\outnode}{\operatorname{out}}

\newcommand{\N}{{\mathbb N}}

\newcommand{\vrt}[1]{V(#1)}

\newcommand{\vrtG}{\vrt{G}}

\newcommand{\edge}[1]{E(#1)}

\newcommand{\edgeG}{\edge{G}}

\newcommand{\AAA}{\mathcal{A}} \newcommand{\BBB}{\mathcal{B}}
\newcommand{\CCC}{\mathcal{C}}

 \newcommand{\LLL}{\mathcal{L}}

\newcommand{\comment}[1]{}
\newcommand{\nwcrown}{nowhere crownful\xspace}
\newcommand{\Nwcrown}{Nowhere crownful\xspace}
\newcommand{\NwCrown}{Nowhere Crownful\xspace}
\newcommand{\swcrown}{somewhere crownful\xspace}

\newcommand{\rcdbg}{(G,\beta,\lambda,\eta)}

\newcommand{\margin}[1]{\marginpar{\smaller{\itshape{#1}}}}

\newcommand{\pprob}[4]{%
\begin{center}\normalfont\fbox{%
\begin{tabular}[t]{rp{#1}}%
\textit{Input:} & #2\\%
\textit{Parameter:} & #3\\%
\textit{Problem:} & #4%
\end{tabular}}%
\end{center}}

\begin{document}

  \title{Directed Nowhere Dense Classes of Graphs}

\author{Stephan Kreutzer\inst{1} \and Siamak Tazari\inst{2}\fnmsep\thanks{supported by a fellowship within the Postdoc-Programme of the German Academic Exchange Service (DAAD).}}

\authorrunning{S. Kreutzer, S. Tazari}

\institute{University of Oxford, \email{kreutzer@comlab.ox.ac.uk}  
  \and Massachusetts Institute of Technology, \email{stazari@mit.edu} }

\date{} 

\maketitle

\begin{abstract}
  We introduce the concept of shallow directed minors and based on
  this a new classification of classes of directed graphs which is
  diametric to existing directed graph decompositions and width
  measures proposed in~the~literature.

  We then study in depth one type of classes of directed graphs which
  we call \emph{nowhere crownful}. The classes are very general as
  they include, on one hand, all classes of directed graphs whose
  underlying undirected class is nowhere dense, such as planar,
  bounded-genus, and $H$-minor-free graphs; and on the other hand,
  also contain classes of high edge density whose underlying class is
  not nowhere dense. Yet we are able to show that problems such as
  directed dominating set and many others become fixed-parameter
  tractable on nowhere crownful classes of directed graphs. This is
  of particular interest as these problems are not
  tractable on any existing digraph measure for sparse classes.

  \smallskip
  The algorithmic results are established via proving a structural
  equivalence of nowhere crownful classes and classes of graphs which
  are \emph{directed uniformly quasi-wide}. This rather surprising
  result is inspired by~[{Ne\v{s}et{\v r}il} and Ossana de Mendez
  2008] and yet a different and much more involved proof is needed,
  turning it into a particularly significant part of our contribution.
\end{abstract}

\section{Introduction}

Faced with the seeming intractability of problems such as variants of
the dominating set problem, the independent set problem and many other
problems naturally arising in applications, an intensively studied
aspect of complexity theory is to explore the boundary of tractability
of these problems by identifying specific classes of graphs on which
they become tractable (in a parameterised setting, see
Section~\ref{sec:prelim}).  A central objective of this research is
to identify structural properties of graphs or graph classes such that
classes of this structure exhibit a rich algorithmic theory --
i.e.~many otherwise intractable graph problems become tractable --
while at the same time these classes should be general enough so that
graphs of this form do occur in applications.

To find such graph parameters, methods derived from structure
theory for undirected graphs have proved to be extremely useful.
Of particular importance in this context is the concept of
\emph{tree-width} (see e.g.~\cite{Diestel05}) developed by Robertson and Seymour as part of
their celebrated graph minor project. Following the introduction of tree-width, a large
number of generally intractable problems have been shown to become
tractable on graph classes with a fixed upper bound on the tree-width. See
\cite{Bodlaender05,Bodlaender98,Bodlaender97,Bodlaender93} for surveys
of tree-width related results.

Besides graph classes of bounded tree-width, many other structural
parameters of graphs have been studied which allow for more efficient
solutions of otherwise hard problems. Among the most important such
parameters are planar graphs, or much more generally, graph classes
excluding a fixed minor (see e.g.~\cite{Diestel05}). A relatively new addition to the family of
graph parameters studied with algorithmic applications in mind are
\emph{nowhere dense} classes of graphs \cite{NesetrilOss08d} which will be of
special importance for this paper. 

The structural parameters discussed above all relate to
undirected graphs. However, many models naturally occurring in
computer science are directed. 
Given the enormous success width parameters 
had for problems defined on undirected graphs, it is natural to ask
whether they can also be used to analyse the complexity of problems
on directed graphs.
 While in principle it is 
possible to apply the structure theory for undirected graphs to directed
graphs by ignoring the direction of edges, this implies a significant
information loss. Hence,
for computational problems which inherently apply to directed graphs, 
methods based on the structure theory for undirected graphs may not 
always be applicable.

Reed \cite{Reed99} and Johnson, Robertson, Seymour and Thomas
\cite{JohnsonRobSeyTho01} initiated the development of a decomposition
theory for directed graphs with the aim of defining an analogue of the
concept of undirected tree-width for directed graphs. Following their
definition of a \emph{directed tree-width}, several alternative
notions have been introduced, for instance in
\cite{Safari05,Barat06,BerwangerDawHunKre06,Obdrzalek06,HunterKre08,KanteR09}.
For each of these decompositions and associated width measures it has
been shown that several problems become tractable if the width of
digraphs with respect to these measures is bounded by a fixed
constant. However, most examples of problems becoming tractable are
either linkage problems, i.e.~problems asking for the existence of
certain pairwise disjoint paths in the digraph, or certain
combinatorial games arising in verification. The only exception is
bi-rank-width \cite{KanteR09} where all problems definable in monadic
second-order logic become tractable. However, this is a very special
case as it is modelled after the concept of undirected clique-width %
and in fact bounded bi-rank-width implies bounded clique-width of the
underlying undirected graphs. It is therefore a measure aimed to some
extent at dense but homogeneous graphs whereas here we are concerned
with sparse graphs in line with undirected graph structure theory such
as excluded minors and nowhere dense classes of graphs.

Following these initial proposals for directed analogues of tree-width
or other undirected graph measures, several papers investigated how
rich the algorithmic theory of classes of digraphs of bounded width
with respect to these measures is.  Unfortunately, for many
interesting problems other than those mentioned before, strong
intractability results for these width measures were obtained showing
that the algorithmic applicability of the existing directed width
measures is very limited. See
e.g.
\cite{KreutzerOrd08,KaouriLamMit08,DankelmannGK09,GanianHliKneLanObdRos10}
and references therein.

\medskip 

\noindent\textbf{Our contributions. }
In this paper we define new and completely different width measures
for directed graphs tailored towards algorithmic applications and which
overcome the problems with existing measures mentioned above.  The
novelty and fundamental difference of our approach is that we will not
consider acyclic digraphs as the algorithmically simplest class of
digraphs, as it was done for previous proposals of directed width
measures, and instead aim for width measures where the class of all
acyclic graphs has unbounded width.

More specifically, on a conceptual level, we draw inspiration from the
concept of nowhere 
dense classes of undirected graphs \cite{NesetrilOss08d} and introduce
a new classification of classes of directed graphs based on the concept
of shallow directed minors (see Section~\ref{sec:diminors}).  In this
way we introduce classes of directed graphs which we call
\emph{nowhere dense}, \emph{somewhere dense}, \emph{\nwcrown} or of
\emph{directed bounded expansion} and show that these classes form a
strict hierarchy into which all classes of directed graphs can be
classified.  We show various structural properties of these classes.

In the rest of the paper we then concentrate on classes of directed
graphs which are {\nwcrown}.  As the main result in the first,
conceptual part of the paper, we show that these classes can
alternatively be characterised by a property called \emph{uniformly
  quasi-wideness} (see Theorem~\ref{thm:main_equivalence}). This
characterisation in terms of wideness properties yields a new and
different perspective of {\nwcrown} classes which is algorithmically
very useful. A similar characterisation for the undirected case was
given in \cite{NesetrilOss10}. Unfortunately, their proof does not
generalise to the directed graph setting and our proof of
Theorem~\ref{thm:main_equivalence} requires completely new and
different ideas.

This equivalence is indeed quite surprising when seen from the following
perspective. The natural analogon to undirected nowhere dense
classes seems to be our notion of directed nowhere dense, where
tournaments are excluded as shallow minors. However, it turns out that
in order to obtain quasi-wideness we need to exclude shallow
\emph{crowns}, a certain orientation of a subdivision of a clique. But
even in this case, the proof of~\cite{NesetrilOss10} breaks down at a
somewhat unexpected place and can only be fixed via a much more
delicate and, in a sense, fragile argument.

\Nwcrown classes are incomparable to classes of digraphs of bounded
width with respect to existing width measures. They can be very
general. For instance, the class of planar directed graphs is
\nwcrown. More generally, if $\CCC$ is a class of directed graphs such
that the class $\undir{\CCC}$ of undirected graphs obtained from
$\CCC$ by forgetting the direction of edges is nowhere dense, e.g.\
excludes a fixed minor, then $\CCC$ is \nwcrown. However, there are
simple examples for classes $\CCC$ of directed graphs which are
{\nwcrown} but where the class $\undir{\CCC}$ is not nowhere
dense. Hence, \nwcrown classes can be much more general than classes
of bounded directed tree-width or other width measures for directed
graphs introduced so far. On the other hand, the class of acyclic
directed graphs is not \nwcrown but it has small width in all other
digraph decompositions (except bi-rank-width) so that \nwcrown classes
are incomparable to classes defined by other width measures.

In the second, algorithmic part of the paper we then concentrate on
algorithmic applications of the graph classes introduced in the first
part.  Our main algorithmic result is that using the alternative
characterisation of {\nwcrown} classes we can show that on these
classes of digraphs problems such as directed (independent, etc. )
dominating set, independent sets, and many other generally hard
problems become tractable. This is particularly interesting as exactly
these types of problems have been shown to be intractable on all
existing directed width measures for sparse classes and therefore our
definition provides the first structural property for sparse classes
that can successfully be used in the analysis of domination problems
on directed graphs. We finally study the connected dominating set
problem and its directed analogues. In particular we show that the
dominating outbranching problem is tractable on {\nwcrown} classes.

Width measures for undirected graphs have proved to be extremely
successful in tackling the complexity of hard algorithmic problems on
undirected graphs with a huge number of applications.  Developing a
similar approach for directed graphs therefore has the potential for
tremendous impact on the theory of hard problems on directed graphs
and on the design of algorithms for solving these and must therefore
be seen as a crucial aim.

The width measures introduced here allow for the first time to analyse
the complexity of domination type problems on classes of directed graphs and
we believe that the structural classification of directed graphs
developed here may provide an interesting step towards
establishing a structure theory for directed graphs more
fruitful for algorithmic applications than the theory based
on analogues of tree-width. 

\medskip

\noindent\textbf{Organisation and results. } The paper is organised as follows.
We fix notation and introduce some basic concepts in
Section~\ref{sec:prelim}. The concept of directed minors used in this
paper is introduced in Section~\ref{sec:diminors} where we establish
basic properties of the minor relation. In
Section~\ref{sec:classification} we establish our classification of
classes of directed graphs by defining several measures for directed
graphs of increasing complexity.

In the rest of the paper we will then concentrate on \nwcrown classes
of graphs. In Section~\ref{subsec:equiv}, we provide an alternative
characterisation of \nwcrown classes based on the existence of
scattered sets. We use this characterisation in Section~\ref{sec:algo}
to show that many problems become tractable on classes of digraphs
which are \nwcrown. We conclude and state open problems in
Section~\ref{sec:conclusion}.

\section{Preliminaries}
\label{sec:prelim}

We write $\N$ for the set of non-negative integers. If $M$ is a set
and $k\in \N$ we write $[M]^{\leq k}$\margin{$[M]^{\leq k}$} for the
set of all subsets of $M$ of cardinality at most
$k$. $[M]^{k},[M]^{<k}$ is defined analogously. 

A \emph{digraph}
$G=(V,E)$ is a pair where $V$ is its set of vertices and $E \subseteq
V \times V$ is its set of edges. We often use $\vrtG$ and $\edgeG$ to
refer to the set of vertices and edges of $G$, respectively, and write
$uv$ instead of $(u,v)$ to denote an edge of $G$. In an
\emph{undirected graph} $G=(V,E)$, we have $E \subseteq [V]^2$ instead.

\begin{definition}
  For a digraph $G$, we define its \emph{underlying undirected graph}
  $\undirG$ \margin{$\undirG$} to be an undirected graph on the same
  vertex set and with edge set $\set{ \set{u,v} \st uv\in E(G) \text{
      or } vu\in E(G) }$. If $\CCC$ is a class of directed graphs,
  then its underlying undirected class is defined as $\undir{\CCC} :=
  \{ \undirG \st G\in \CCC\}$.

  Similarly, for an undirected graph $G$, we define its
  \emph{underlying bidirected graph} $\bidirG$ \margin{$\bidirG$} as
  the directed graph on the same vertex set and with edge set $\set{
    (u,v), (v,u) \st \set{u,v} \in E(G) }$ and define $\bidir{\CCC}$
  for a class of undirected graphs $\CCC$ analogously.
\end{definition}

A \emph{(directed) path} of length $k$ in a (di)graph $G$ is a
sequence $v_1\dots v_{k+1}$ of distinct vertices of $G$ such that for
each $1 \leq i \leq k$, there is an edge $v_iv_{i+1}$ in $E(G)$. If we
allow and require $v_1=v_{k+1}$ we have a \emph{(directed) cycle} instead. A
\emph{directed acyclic graph (DAG)} is a digraph that contains no
directed cycles.

\begin{definition}
  By $N_d^+(v)$\margin{$N_d^+(v)$} we denote the \emph{$d$-outneighbourhood} of $v$, i.e.~
  \[
     N_d^+(v) := \{ u\in \vrtG \st \text{ there is a directed path
       from $v$ to $u$ of length } \leq d\}.
  \]
  $N_d^-(v)$ is defined analogously as the \emph{$d$-inneighbourhood}
  of $v$. \margin{$N_d^-(v)$} 

  If $X \subseteq V(G)$ then $N_d^+(X) := \bigcup_{x \in X} N_d^+(x)$,
  and $N_d^-(X)$ is defined analogously. We skip the index $d$ when it
  is equal to $1$.
\end{definition}

\begin{definition}
  A \emph{directed bipartite graph} is a directed graph $G := (A
  \dot\cup B, E)$ whose vertex set is partitioned into two sets $A$
  and $B$ and $E\subseteq A\times B$.  
%
\end{definition}

An \emph{orientation} of an undirected graph $G$ is a directed graph
obtained from $G$ by replacing every undirected edge by a directed
edge. A \emph{clique} of order $n$ is the (up to isomorphism) complete
undirected graph $K_n$ on $n$ vertices containing all possible
edges. A \emph{tournament} of order $n$ is a digraph $T_n$ that is an
orientation of the clique $K_n$.

A \emph{subdivision} of a (di)graph $G$ is obtained by replacing some
edges of $G$ by paths; in case of digraphs, the paths must respect the
directions of the replaced edges.  The following class of graphs, which
are acyclic orientations of subdivided cliques, will play a special
role in this paper.

\begin{definition}
  A \emph{crown} of order $q$, for $q>0$, is the graph $S_q$ with \margin{$S_q$}
  \begin{itemize}
  \item $V(S_q) := \set{ v_1, \dots, v_q \} \dot\cup \{ u_{i,j} \st 1
    \leq i < j\leq q}$ and
  \item $E(S_q) := \set{ (u_{i,j}, v_i),(u_{i,j}, v_j) \st 1\leq i<j\leq
    q}$.
  \end{itemize}
  we call the vertices $v_1, \dots, v_n$ the \emph{principal vertices}
  of the crown.
\end{definition}

\begin{definition}
  Let $G$ be a digraph. A \emph{$k$-alternating path} $AP_k$ in $G$,
  for some $k\geq 1$, is an orientation of a path $v_1\dots v_{k+2}$
  such that either all edges are directed towards their incident
  vertex with an odd index or all edges are directed towards their
  incident vertex with an even index.
\end{definition}

The model of complexity we are using is parameterised complexity
\cite{DowneyFel98,FlumGro06}. A problem of size $n$ with parameter $k$
is said to be \emph{fixed-parameter tractable (fpt)}, if it can be
solved by an algorithm in time $\Oof(f(k) n^{\Oof(1)})$, for some
computable function $f$. The class FPT is the set of all parameterised
problems that are fixed-parameter tractable. The class XP is the set
of all parameterised problems that can be solved by an algorithm in
time $\Oof(n^{f(k)})$, for a computable function $f$. If the parameter
of a problem is not explicitly specified, we assume the \emph{standard
  parameterization}, that is, the parameter is the solution size. For
example, the standard parameter of the dominating set problem is the
size of a minimum dominating set in the given graph. Recall that $D
\subseteq V(G)$ is a $d$-dominating set if $N_d^+(D) = V(G)$
and is just called a dominating set for $d=1$.

\section{A Classification of Directed Graph Classes}

In this section we give a classification of directed graphs in terms
of shallow directed minors. We first state our quite general
definition for directed minors in terms of models and compare it with
some other definitions in the literature. In
Section~\ref{sec:classification}, we define several properties of
directed graph classes, show their relationship, and compare them to their
undirected counterparts. 

\subsection{Directed Minors}
\label{sec:diminors}
For undirected graphs, we say that $H$ is a minor of $G$, denoted as
$H \minor G$ \margin{$\minor$}, if $H$ can be obtained from $G$ by a series of vertex
and edge deletions and edge contractions. This is equivalent to $G$
containing a \emph{model} of $H$. We define the minor relation for
directed graphs in terms of models below.

\begin{definition}\label{def:dminor}
  A digraphs $H$ has a \emph{directed model} in a digraph $G$ if there is
  a function $\delta$ mapping vertices $v\in V(H)$ of $H$ to
  sub-graphs $\delta(v) \subseteq G$ and edges $e\in E(H)$ to edges
  $\delta(e) \in E(G)$ such that 
  \begin{itemize} \squishlist
    \item if $v\not= u$ then $\delta(v) \cap \delta(u) = \emptyset$;
    \item if $e := uv$ then $\delta(e)$ has its start point in
      $\delta(u)$ and its end in $\delta(v)$.
  \end{itemize}
  For $v\in V(H)$ we set $\innode(\delta(v)) := V(\delta(v)) \cap
  \bigcup_{e := uv\in E(H)} V(\delta(e))$ and\\
  $\outnode(\delta(v))
  := V(\delta(v)) \cap \bigcup_{e := vw\in E(H)} V(\delta(e))$ and
  require that 
  \begin{itemize} \squishlist
    \item there is a directed path in $\delta(v)$ from any $u\in \innode(\delta(v))$ to
      every $u'\in \outnode(\delta(v))$;
    \item there is at least one source vertex $s_v \in \delta(v)$ that
      reaches every element of $\outnode(\delta(v))$;
    \item there is at least one sink vertex $t_v \in \delta(v)$ that can
      be reached from every element of $\innode(\delta(v))$.
  \end{itemize}
  We write $H\dminor G$ if $H$ has a directed model in $G$ and say $H$
  is a \emph{directed minor} of $G$.\margin{$\dminor$} We call the
  sets $\delta(v)$ for $v \in V(H)$ the \emph{branch-sets} of the
  model.
\end{definition}

It is obvious that this notion generalises the concept of undirected
minors as follows.

\begin{lemma}\label{lem:minor_dminor}
  If $G$, $H$ are undirected graphs, then $H \minor G \Leftrightarrow
  \bidir H \dminor \bidir G$. If $G$, $H$ are digraphs, then $H
  \dminor G \Rightarrow \undir H \minor \undir G$.
\end{lemma}

For (di)graphs $G$, $H$, we say that $H$ is a \emph{(directed)
  topological minor} of $G$ if $G$ contains a subdivision of $H$ as a
subgraph and denote it by $H \tminor G$ \margin{$\tminor$}. As with
undirected graphs, we have that $H \tminor G \Rightarrow H \dminor G$
but the reverse is not true in general.

In the literature, another notion of directed minor has been considered~\cite{JohnsonRobSeyTho01}:

\begin{definition}
  A \emph{butterfly contraction} is the operation of contracting an
  edge $e=uv$ where either $u$ has outdegree $1$ or $v$ has indegree
  $1$. A graph $H$ is said to be a \emph{butterfly minor} of $G$,
  i.e.\ $H \bminor G$\margin{$\bminor$}, if it can be obtained
  from $G$ by a series of vertex and edge deletions and butterfly
  contractions.
\end{definition}

\begin{figure}[t]
  \centering
  \subfigure[]{\includegraphics[height=1cm]{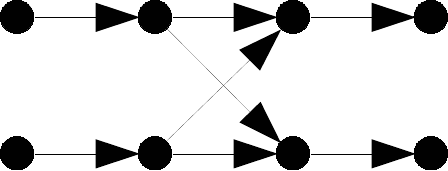}}\hfil\hfil
  \subfigure[]{\includegraphics[height=1cm]{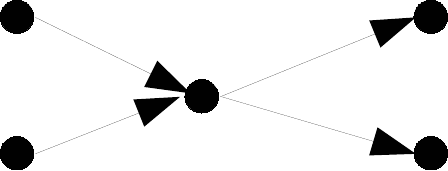}}\hfil\hfil
  \caption{The digraph in (a) contains a model of (b)
    but does not contain it as a butterfly minor.}
  \label{fig:bminor} 
\end{figure}

\begin{lemma}
  For digraphs $G$ and $H$, we have $H \bminor G \Rightarrow H\dminor G$,
  but the converse is not true in general.
\end{lemma}
\begin{proof}
  We prove the claim by induction on the number of butterfly
  contractions necessary to obtain $H$ from $G$.  If no contractions
  are necessary, there is nothing to show. Otherwise, let $e=uv$ be an
  edge of $G$ that is to be butterfly contracted to a vertex $z$ in
  $G':=G/e$. Consider any path in $G'$ that passes through $z$ via
  edges $e_1$ and $e_2$, in this order. If $e_1$ is incident to $v$ in
  $G$ and $e_2$ is incident to $u$ in $G$, then $v$ would have
  indegree $2$ and $u$ would have outdegree $2$ in $G$, making it
  impossible to butterfly contract $e$ in $G$. Hence, every path in
  $G'$ corresponds to a path in $G$ by uncontracting $e$.

  As $H$ is a butterfly minor of $G'$, we know by the induction
  hypothesis that there is a map $\delta'$ that describes a model of
  $H$ in $G'$. Let $z_H \in V(H)$ be such that $z \in
  \delta'(z_H)$. Define $\delta$ to be equal to $\delta'$ except for
  $\delta(z_H) := \delta'(z_H) \setminus \set{z} \cup \set{u,v}$. Then
  $\delta$ describes a valid model of $H$ in $G$ since all paths of
  $G'$ are preserved when uncontracting $e$.

  To see that the converse is not true in general, see
  Figure~\ref{fig:bminor}. \qed
\end{proof}

In the rest of the paper, models of directed bipartite graphs will play
a special role. For these, we can take models to be of a particularly
simple form. Recall that an in-branching is an orientation of a tree
with all edges oriented towards the root; an out-branching is defined
analogously.

\begin{lemma}
  If $H$ is a directed bipartite graph with $H \dminor G$, we can
  choose the branch-sets of the model of $H$ in $G$ to be in- or
  out-branchings. Furthermore, we have $H \dminor G \Leftrightarrow H
  \bminor G$.
\end{lemma}
\begin{proof}
  The first claim follows immediately from the definition of a
  directed model since $H$ consists solely of sources and sinks. As
  for the second claim, note that in order to obtain $H$ from $G$ as a
  butterfly minor, we can first delete all edges that are not in the model
  of $H$ in $G$ and then butterfly contract the in- and out-branchings
  of the model starting from the leaves.\qed
\end{proof}


If $G$ is acyclic, we can detect fixed minors in polynomial time (more
precisely, XP-time w.r.t.\ $|H|$). In order to prove this
fact, we first show the following lemma.

\begin{lemma}\label{lem:disjointpaths}
  Let $G$ be a DAG and let $s_1,\dots,s_k,t_1,\dots,t_k$ be given
  vertices of $G$. Furthermore, let $\LLL=\set{(z_0=0,z_1],
    (z_1,z_2],\dots,(z_{\ell-1},z_{\ell}=k]}$ be a set of $\ell$ given
  intervals partitioning $[1,k]$. The goal is to find, for each $1
  \leq i \leq k$, a path $P_i$ between $s_i$ and $t_i$ such that if $1
  \leq i < j \leq k$ belong to different intervals of $\LLL$, then
  $P_i$ and $P_j$ are disjoint.
  \begin{itemize} \squishlist
    \item[(a)] There is an algorithm running in time $n^{\Oof(k)}$ that
      finds such paths if they exist.
    \item[(b)] If we are additionally given an integer $r$ to bound
      the length that each path is allowed to have, then in time
      $(nr)^{\Oof(k)} = n^{\Oof(k)}$ we can find the required paths or report that
      they do not exist.
    \end{itemize}
\end{lemma}
\begin{proof}
  Our ideas generalise that of Eppstein's algorithm~\cite{Eppstein95} (see
  also Perl and Shiloach~\cite{PerlShiloach78}) for finding disjoint
  paths in DAGs. We first augment $G$ by a vertex $s$ and add edges
  $(s,s_i)$ for all $1 \leq i \leq k$. Then we construct a new DAG $D$
  that contains a vertex for each tuple $(v_1,\dots,v_k) \in
  V(G)^k$. Consider a topological order of the vertices of $G$ and let
  $f(v)$ denote the position of vertex $v$ in this order. We add an
  edge from a tuple $(v_1,\dots,v_k)$ to a tuple $(w_1,\dots,w_k)$ in
  $D$ if the following conditions apply for some $w \in
  \set{w_1,\dots,w_k}$ and $d \in [1,\ell]$:
  \begin{itemize} \squishlist
    \item[(i)] for $1 \leq i \leq k$, if $w_i=w$, then there exists an edge $(v_i,w_i) \in E(G)$;
    \item[(ii)] for $1 \leq i \leq k$, if $w_i \neq w$, then $v_i=w_i$, i.e.\ $w$ is the only new vertex introduced;
    \item[(iii)] for $i \not\in (z_{d-1},z_d]$, $v_i=w_i$, i.e.\ we may change only elements out of one interval of $\LLL$ at a time;
    \item[(iv)] for $1 \leq i \leq k$, $f(w) > f(v_i)$.
  \end{itemize}
  We claim that there is a one-to-one correspondence between the paths
  $P_i$ that we are looking for in $G$ and paths from $(s,\dots,s)$ to
  $(t_1,\dots,t_k)$ in $D$. Indeed, we can transform the latter to the
  former by taking the projection to the $i$th coordinate for each
  $i$; condition (iii) and (iv) above guarantee that paths from
  different intervals will be disjoint since $w$ is required to be a
  vertex that has not appeared on the path so far and may be used by
  only one interval. In order to transform paths in $G$ to a path in
  $D$, we consider the set of all vertices appearing on the paths
  $P_1,\dots,P_k$ and introduce them to our path in $D$ in topological
  order. For further details, we refer to~\cite{Eppstein95}\footnote{the proof
  therein is for the all disjoint case but can be easily adapted.}.

  In order to prove (b), we additionally keep a table of size $(r+1)^k$ at
  each node of $D$ indexed by tuples $(r_1,\dots,r_k) \in
  \set{0,\dots,r}^k$. At a vertex $v \in V(D)$, an entry
  $(r_1,\dots,r_k)$ should indicate whether there exists a path $P$ from
  $(s,\dots,s)$ to $v$ in which the length of the path in $G$
  corresponding to the projection of the $i$th component of $P$ has
  length $r_i$. If we traverse the vertices of $D$ in topological order,
  these tables can be filled easily at each vertex using the tables of
  its predecessors, resulting in the claimed overall running time.\qed
\end{proof}

\begin{theorem}\label{lem:dminor_alg}
  Let $G$ be a directed acyclic graph. There is a polynomial time
  algorithm to decide whether $H \dminor G$ for any fixed digraph
  $H$.
\end{theorem}
\begin{proof}
  Obviously, if $H$ contains a cycle, we can reject. Otherwise, assume
  $G$ has $n$ vertices and $H$ has $h$ edges. We try to construct a
  map $\delta$ that defines a model of $H$ in $G$ as follows. First,
  we guess the mapping of the edges of $H$, that is, for each $e \in
  E(H)$, we guess an edge $e' \in E(G)$ as the value of
  $\delta(e)$. Furthermore, for each $v \in V(H)$, we guess vertices
  $s_v,t_v \in V(G)$ that act as a source and a sink for $\delta(v)$,
  respectively. There are $n^{\Oof(h)}$ possibilities for these
  choices and we can try all of them. This way, we have also partially
  guessed some vertices of $\delta(v)$, for $v \in V(H)$, namely, all
  vertices of $\innode(\delta(v)), \outnode(\delta(v))$, and its
  source and sink vertices. What remains is to connect these vertices
  by some paths as specified in Definition~\ref{def:dminor} such that
  the paths belonging to $\delta(v)$ and $\delta(w)$ for distinct
  $v,w \in V(H)$ are disjoint. But this is exactly what we can check
  for using Lemma~\ref{lem:disjointpaths} in $n^{\Oof(h^2)}$ time,
  which is polynomial for constant $h$ (note that the number of paths
  we are looking for can be $\Oof(h^2)$).\qed
\end{proof}

A further concept that is central to our work is that of \emph{shallow} minors:

\begin{definition}
  A digraph $H$ is a \emph{(shallow) depth $r$ minor} of a digraph
  $G$, denoted as $H \dminor_r G$, if there exists a directed model of
  $H$ in $G$ in which the length of all the paths in the branch-sets
  of the model are bounded by $r$.
\end{definition}

Hence, a depth $0$ minor is simply a subgraph and any minor of $G$ is
a depth $n$ minor. This is analogous to the definition of depth $r$
minors in undirected graphs, denoted by $H \minor_r G$. 

\begin{theorem}
  Let $G, H$ be digraphs and $r$ an integer. There is an algorithm
  that decides whether $H \dminor_r G$ in time $n^{\Oof(rh^2)}$,
  where $h$ denotes the number of edges of $H$. If $G$ is acylcic,
  then the problem can be decided in time $(nr)^{\Oof(h^2)} = n^{\Oof(h^2)}$.
\end{theorem}
\begin{proof}
  As in the proof of Lemma~\ref{lem:dminor_alg}, we start by guessing
  the image of the edges of $H$ in $G$ and then check if the required
  paths to complete the model exist. But as there are at most $n^r$
  paths of length at most $r$ starting at any given vertex of $G$, we
  can accomplish this in total time $n^{\Oof(rh^2)}$.

  If $G$ is a DAG, we can search for the required paths using Lemma~\ref{lem:disjointpaths}~(b). \qed
\end{proof}


\subsection{Classes of Digraphs Excluding Shallow Minors}
\label{sec:classification}


In the realm of undirected graphs, classes of graphs excluding minors
have been studied extensively, on one hand because some of these
classes, like planar or bounded genus graphs, are very natural and
important, and on the other hand due to their structural richness
revealed by Robertson and Seymour's graph minor
theory~\cite{GM-series}. More recently, {Ne\v{s}et{\v r}il} and Ossana
de Mendez~\cite{NesetrilOss08d} offered a generalization of these
classes based on the notion of shallow minors. They define a class of
undirected graphs to be \emph{somewhere dense} if there exists an
integer $r$, such that the set of depth $r$ minors of $\CCC$ contains
arbitrarily large cliques; otherwise, the class is called
\emph{nowhere dense}. It turns out that nowhere dense classes of
graphs are the most general classes of (almost) sparse graphs that are
structurally still very rich and allow for efficient (parameterised)
algorithms~\cite{DawarKre09,NesetrilOss10}.

In order to obtain an analogous theory for digraphs, we use the notion
of shallow directed minors as introduced above. One of our design
goals is to obtain proper generalizations of the undirected classes in
the sense that is given in our Classification and Comparison
Theorems~\ref{thm:classification} and~\ref{thm:comparison}
below. However, from an algorithmic point of view, we face the
following important issue: every acyclic graph excludes any cycle as a
minor, and every directed bipartite graph excludes a path of length
$2$ as a minor. But acyclic graphs are already algorithmically very
hard for many problems; and directed bipartite graphs (together with a
super-source) easily encode, for example, the general dominating set
problem in undirected graphs, which is W[2]-hard. Hence it seems that
even excluding an acylic graph will not facilitate finding efficient
algorithms.  Therefore, it appears that, algorithmically, it might be
beneficial to additionally introduce and study classes of digraphs
excluding a bipartite graph as a (shallow) minor.

\begin{definition}
  Let $\CCC$ be a class of digraphs. 
  \begin{enumerate}
  \item $\CCC$ is \emph{directed somewhere dense} if there is a radius
    $r\geq 0$ so that the set of depth $r$ minors of $\CCC$ contains
    arbitrarily large tournaments.
  \item $\CCC$ is \emph{directed nowhere dense} if for every $r$,
    there exists an $n$ and an \emph{acyclic} tournament $T_n$ so that
    for all $G\in\CCC$, $T_n \not\dminor_r G$.
  \item $\CCC$ is \emph{\swcrown}if there is a radius $r
    \geq 0$ so that every crown $S_q$ occurs as a depth $r$ minor of a
    graph $G \in \CCC$.
  \item $\CCC$ is \emph{\nwcrown}if for every $r$, there
    exists a $q$ so that for all all $G \in \CCC$, $S_q \not\dminor_r
    G$.
  \item $\CCC$ has \emph{directed bounded expansion} if there is a
    function $f\st \N\rightarrow \N$ such that $\nabla^d_r(G) \leq
    f(r)$ for all $G\in \CCC$, where $\nabla^d_r(G) := \{
    \frac{|E(H)|}{|V(H)|} \st H \dminor_r G\}$ (cf.~\cite{NesetrilOss08}).
  \item $\CCC$ is \emph{crown-minor-free} if there exists a
    $q$ such that $S_q \not\dminor G$ for any $G \in \CCC$.
  \item $\CCC$ is \emph{alternating path-minor-free} if there exists a
    $k$ such that $AP_k \not\dminor G$ for any $G \in \CCC$.
  \end{enumerate}
\end{definition}

\begin{theorem}[Classification Theorem]\label{thm:classification}~
\begin{itemize} \squishlist
  \item[(a)] A class $\CCC$ is directed somewhere dense if and only if
    it is not directed nowhere dense, i.e.\ there exists a radius $r$
    so that every acyclic tournament occurs as a depth $r$ minor of a
    graph $G \in \CCC$.
  \item[(b)] The property of being directed nowhere dense is more
    general than being \nwcrown, which is in turn more general
    than being crown-minor-free, and the latter being in turn more
    general than being alternating-path-minor-free.
  \item[(c)] The property of being directed nowhere dense is also more
    general than being of directed bounded expansion, which is in turn
    more general than being alternating-path-minor-free. The property
    of being of directed bounded expansion is neither comparable to
    being crown-minor-free, nor to being \nwcrown.
\end{itemize}
\end{theorem}
\begin{proof}
  (a) This follows from the fact that any tournament of order $2^n$
  contains all acyclic tournaments of order $n$ as a subgraph. For,
  let $T$ be a tournament of order $2^n$ and $T'$ an acyclic
  tournament of order $n$. Let $v$ be a vertex of highest outdegree in
  $T$ and let $v'$ be the first vertex of $T'$ in a topological
  order. Let $T_v$ be the tournament induced by the out-neighbours of
  $v$ in $T$. As the out-degree of $v$ is at least $2^{n-1}$, we have
  by induction that there exists a copy of $T'-v'$ in $T_v$. Now
  matching $v$ with $v'$ (note that $v'$ can only have out-neighbours
  in $T'$) and adding it to this copy gives us a copy of $T'$ in $T$.

  Claim (b) is immediate by definition.

  (c) A class that is of directed bounded expansion cannot have
  arbitrarily dense graphs as depth $r$ minors, in particular all
  acyclic tournaments, and hence any such class is directed nowhere
  dense. Also, as graphs of bounded treewidth and their minors are
  sparse, it follows from Theorem~\ref{thm:alternating} that
  alternating-path-minor-free classes are of directed bounded
  expansion. In Theorem~\ref{thm:density}, we show that there
  exist crown-minor-free classes of digraphs that have high edge
  density -- and are hence not of bounded expansion. Conversely, the
  set of all crowns and their minors (which are only their subgraphs)
  have edge density at most $2$ and are hence of bounded expansion but
  clearly not \nwcrown.\qed
\end{proof}

This theorem shows that our classes are robust; the next theorem shows
that they properly generalise the respective undirected concepts and
capture  much larger  classes of digraphs than what we would
get by only considering their underlying undirected classes.

\begin{theorem}[Comparison Theorem]\label{thm:comparison}
\begin{itemize} \squishlist
  \item[(a)] If $\CCC$ is a class of undirected graphs, then $\CCC$ is
    somewhere dense/nowhere dense/of bounded expansion/$H$-minor-free if
    and only if $\bidir \CCC$ is directed somewhere dense/nowhere
    dense/of bounded expansion/crown-minor-free,
    respectively. Furthermore, $\bidir \CCC$ is directed somewhere dense
    if and only if it is \swcrown.
  \item[(b)] Let $\CCC$ be a class of digraphs. If the underlying
    undirected class $\undir \CCC$ is nowhere dense/of bounded
    expansion/$H$-minor-free, then $\CCC$ is \nwcrown/of directed
    bounded expansion/crown-minor-free, respectively.
  \item[(c)] The reverse of (b) is not true in general. There exist
    classes of digraphs that are crown-minor-free but whose underlying
    undirected class is somewhere dense. 
\end{itemize}
\end{theorem}
\begin{proof}
  (a,b) All claims are immediate by Lemma~\ref{lem:minor_dminor} and
  the fact that the underlying undirected graph of a crown can be
  contracted to a complete minor.

  (c) Consider the class $\CCC := \set{S'_q | q \geq 1}$ where $S'_q$
  is obtained by reversing the direction of all the edges of
  $S_q$. Since $S'_q$ is a directed bipartite graph, it has no
  directed minors except for its subgraphs; but since the indegree of
  each vertex is at most $2$, it does not contain any crown of order
  greater than $3$ as a subgraph and is hence crown-minor-free. But
  $\undir \CCC$ obviously contains all cliques as minors and is hence
  somewhere dense. \qed
\end{proof}

The previous theorem already indicated that the concepts of directed
graphs introduced here properly generalise their undirected
counterparts. The following theorem gives further evidence of this,
where we show that that crown-minor-free graphs can in fact be quite
dense -- something that classes of undirected graphs which are nowhere
dense cannot be.

\begin{theorem}\label{thm:density}
  For every $\epsilon$, there exists a $q$ and
  an $S_q$-minor-free digraph that has edge density at least
  $\Omega(n^{\frac12 - \epsilon})$ (whereas undirected nowhere dense
  classes have edge density at most
  $n^{o(1)}$~\cite{NesetrilOss08d}).
\end{theorem}
\begin{proof}
  We use a probabilistic construction to obtain a directed
  bipartite graph $G = (A \dot\cup B, E)$ with $|A|=|B|=n$ as
  follows. For a given density $d := d(\epsilon)$, we choose $d$
  neighbours for every vertex $a \in A$ uniformly at random. This way,
  our graph $G$ has exactly $dn$ edges and since it is a directed
  bipartite graph it has no directed minors except for its
  subgraphs. Now, let us count the expected number of $S_q$'s in this
  graph. For a set $\beta=\set{b_1,\dots,b_q} \subseteq B$ and an
  ordered sequence $\alpha=(a_{12},\dots,a_{(q-1)q}) \in
  A^{\binom{q}{2}}$, let $X^\alpha_\beta$ be the indicator variable
  that $\alpha \cup \beta$ forms an $S_q$ such that $a_{ij}$ is
  adjacent to $b_i$ and $b_j$. By assuming $d < n/2$, we have

  $$
  P[X^\alpha_\beta = 1] =
  \big(\frac{\binom{n}{d-2}}{\binom{n}{d}}\big)^{\binom{q}{2}} = \big(
  \frac{d(d-1)}{(n-d+2)(n-d+1)} \big) ^{\binom{q}{2}} \leq \big(
  \frac{2d}{n} \big)^{q(q-1)}
  $$
  
  \noindent By linearity of expectation, and choosing $d
  :=\floor{\frac{n^{\frac12 - \epsilon}}{2}}$ and $q > 1 +
  \frac{1}{\epsilon}$, we obtain:
  
  \begin{align*}
    E[\sum_{\alpha,\beta} X^\alpha_\beta] &= \sum_{\alpha,\beta} E[X^\alpha_\beta] \leq \binom{n}{q} \binom{n}{\binom{q}{2}} \binom{q}{2}! \cdot \big( \frac{2d}{n} \big) ^{q(q-1)}\\
    &\leq n^q \cdot n^{\frac{q(q-1)}{2}} \cdot \frac{(2d)^{q(q-1)}}{n^{q(q-1)}} 
    = \frac{(2d)^{q(q-1)}}{n^{\frac{q(q-3)}{2}}}\\
    &\leq n^{(\frac12 - \epsilon)q(q-1) - \frac{q(q-3)}{2}} = n^{q - \epsilon q (q-1)} < 1
  \end{align*}
  Hence, there exists a graph $G$ with edge density $d =
  \Omega(n^{\frac12 - \epsilon})$ that contains no
  $S_q$. \qed
\end{proof}

One might notice that both the examples exhibited in the proofs of
Theorems~\ref{thm:comparison}~(c) and~\ref{thm:density} contain graphs
with very long alternating path. One might think that alternating
paths are a rather special case. But in fact, it is exactly the
existence of such paths that make directed graphs complicated, as
shown in the next theorem.  The main idea is that if $\CCC$ has large
tree-width, then it contains large grid minors which in turn contain
large alternating paths, independent of how they are oriented.

\begin{theorem}\label{thm:alternating}
  If $\CCC$ is a class that is alternating path-minor-free
  then $\CCC$ (or, more precisely, $\undir{\CCC}$) has bounded treewidth.
\end{theorem}
\begin{proof}
  For $l_1,l_2\geq 1$, the $l_1\times l_2$-grid is the undirected
  graph with vertex set $\{ (i,j) \st 1\leq i \leq l_1, 1 \leq j \leq
  l_2\}$ and edge set $\{ \big( (i,j), (i', j') \st |i-i'|
  +|j-j'|=1\}$.

  We claim that if $H$ is an orientation of a $2l \times 3$-grid then
  $H$ contains a subdivision of an $l$-alternating path beginning at
  the vertex $(1, 1)$ and terminating either at $(2l,1)$ or $(2l,3)$.

  The claim is proved by induction on $l$. For $l=1$, let $H$ be an
  orientation of a $2\times 3$-grid. Consider the path
  $P_1 := (1,1)(1,2)(1,3)(2,3)(2,2)(2,1)$; if it contains no alternations,
  then the path $P_3 := (1,1)(1,2)(2,2)(2,3)$ must have at least one
  alternation.

  Now let $l > 1$ and consider the paths $P_1$ and $P_3$ described
  above. By the induction hypothesis, there exists a path $P'_1$
  starting at $(3,1)$ and ending in $(2l, z_1)$ containing at least
  $l-1$ alternations; by symmetry, there exists also a path $P'_3$
  starting at $(3,3)$ and ending in $(2l,z_3)$ containing at least
  $l-1$ alternations (where $z_1,z_3 \in \set{1,3}$). Now either the
  path $P_1P'_1$ or the path $P_3P'_3$ fulfill our claim.

  Now suppose $G$ is a graph of tree-width $\tw(G) \geq f(2k)$. Then,
  by the excluded grid theorem~\cite{GMV}, $G$ contains a $2k\times
  2k$-grid as a minor and therefore a $k$-alternating
  path.\footnote{Note that by routing the path through a model, we can
    only get more alternations.} \qed
\end{proof}

\begin{corollary}
  Let $\CCC_k$ be the class of $AP_k$-free digraphs. Then sub-graph
  isomorphism (and every other problem that is tractable on bounded
  treewidth graphs) is fixed-parameter tractable on $\CCC_k$.
\end{corollary}

Note that the sub-graph isomorphism problem is of particular interest,
as it is as hard on crown-minor-free classes as it is in the general
undirected case: we can transform an undirected instance to a
crown-minor-free instance simply by replacing every edge with an
alternating path of length $2$. The same is true for classes with
directed bounded expansion and even bounded directed path- or
tree-width.

\section{An Alternative Characterisation of \NwCrown Classes}
\label{subsec:equiv}

In this section, we introduce the notion of directed uniformly
quasi-wide classes as a generalization of undirected uniformly
quasi-wide classes introduced by Dawar~\cite{Dawar10}. Our main
purpose is to prove Theorem~\ref{thm:main_equivalence} stating that
this (rather algorithmic and seemingly unrelated) concept is exactly
equivalent to the concept of \nwcrown classes of digraphs. This is
similar to the theorem of {Ne\v{s}et{\v r}il} and Ossana de
Mendez~\cite{NesetrilOss10} showing that undirected nowhere dense
classes are equivalent to uniformly quasi-wide classes. However, their
proof does not generalise to the directed case, and hence we
introduce significantly different ideas and a much more involved
analysis to obtain this result.

\begin{definition}
  Let $G$ be a digraph and $d\in\N \cup \set{0}$. A set $U\subseteq \vrtG$ is
  \emph{$d$-scattered} if there is no $v\in\vrtG$ and $u_1\not=u_2\in
  U$ with $u, u'\in N_d^+(v)$.
\end{definition}

\noindent Note that as $N_0^+(v) = \set{v}$, any subset of $\vrtG$ is
$0$-scattered.

\begin{definition}
  A class $\CCC$ of directed graphs is \emph{uniformly quasi-wide} if
  there are functions $s\st\N\rightarrow\N$ and $N\st
  \N\times\N\rightarrow \N$ such that for every $G\in\CCC$ and all $d,
  m\in\N$ and $W\subseteq \vrtG$ with $|W|>N(d,m)$ there is a set
  $S\subseteq\vrtG$ with $|S|\leq s(d)$ and $U\subseteq W$ with $|U| =
  m$ such that $U$ is $d$-scattered in $G - S$. $s, N$ are called the
  margin of $\CCC$.
    
  If $s$ and $N$ are computable then we call $\CCC$ effectively
  uniformly quasi-wide.
\end{definition}

Note that the class of reversed crowns described in the proof of
Theorem~\ref{thm:comparison}~(c) is an example of a class that
is directed uniformly quasi-wide but whose underlying undirected graph
is not uniformly quasi-wide. The fact that the directed case properly
generalises the undirected case also follows from
Theorem~\ref{thm:main_equivalence} and the equivalence of undirected
nowhere dense classes with undirected uniformly quasi-wide classes.

%
%

\begin{theorem}\label{thm:main_equivalence}
  A class $\CCC$ of digraphs is \nwcrown if and only if it
  is directed uniformly quasi-wide.
\end{theorem}

The if direction is simple.

\begin{lemma}
  If $\CCC$ is uniformly quasi-wide then it is \nwcrown.
\end{lemma}
\begin{proof}
  Let $s, N$ be the margin of $\CCC$.  Towards a contradiction,
  suppose there is $r>0$ such that for all $q$ there is $G_q\in\CCC$
  with $S_q\dminor_r G_q$. Let $s := s(2r+1)$, $N = N(2r+1, s+2)$, and
  let $G\in \CCC$ be such that $S_{N+1} \dminor_r G$. W.l.o.g.~we
  assume $N\geq s+2$.

  Then there are out-branchings $\AAA := \set{A_1, \dots,
    A_{\binom{N+1}{2}}}$ and in-branchings $\BBB=\set{B_1,\dots,
    B_{N+1}}$ witnessing $S_{N+1} \dminor_r G$. Let $W$ be the set of
  principal vertices of the in-branchings. As $|W|>N$ there is a set
  $S\subseteq V(G)$ with $|S|\leq s$ and $U\subseteq W$ with $|U| =
  s+2$ such that $U$ is $(2r+1)$-scattered in $G - S$.

  Let $\BBB' := \set{B'_1,\dots,B'_t} \subseteq \BBB$ be such that
  each $B'_i \in \BBB'$ contains a vertex of $U$ but no vertex of $S$;
  obviously, we have $t \geq 2$. Also, note that at most $t-2$
  elements out of $\AAA$ can be hit by $S$. But since $\binom{t}{2} >
  t-2$ for $t \geq 2$, there must exist an $A_k \in \AAA$ and $B'_i,
  B'_j \in \BBB'$ none of which are hit by $S$. But then $G-S$
  contains a path of length at most $2r+1$ from the root of $A_k$ to
  the principal vertices of $B'_i$ and $B'_j$ contradicting the
  assumption that $U$ is $(2r+1)$-scattered.\qed
\end{proof}

In the remainder of this section we prove the converse. For this, we
first introduce some auxiliary concepts.

\begin{definition}\label{def:rcdbg}
  \begin{itemize}
  \item[(a)] An \emph{$r$-controlled directed bipartite graph} is a
    tuple $\rcdbg$ such that
    \begin{itemize} \squishlist
      \item $G := (A \dot\cup B, E)$ is a directed bipartite graph;
      \item $\beta \st A\rightarrow [B]^{\leq 1}$ is a function
        assigning a \emph{base} $\beta(a) \in B$ to some vertices $a \in
        A$;
      \item $\lambda \st A \rightarrow \set{0,\dots,r+1}$ assigns a
        \emph{level} to each vertex in $A$; and
      \item $\eta \st E \rightarrow [A]^{\leq r+1}$ is a function that
        specifies for each $e=ab \in E$ a subset $\eta(e) \in
        [A]^{=\lambda(a)}$ s.t.\ if $\eta(e) = \set{a_0,\dots,a_i}$
        then $\beta(a_0) = \dots = \beta(a_i) = b$ and $\lambda(a_0) <
        \lambda(a_1) < \dots < \lambda(a_i)$; furthermore, if
        $\beta(a) = b$, we have $\lambda(a_i) < \lambda(a)$.
    \end{itemize}
  \item[(b)] A \emph{controlled crown} of order $q$ in an
    $r$-controlled directed bipartite graph $\rcdbg$ is a crown $S_q
    \subseteq G$ such that for all $e \in E(S_q)$ we have $\eta(e)
    \cap V(S_q) = \emptyset$.
  \end{itemize}
\end{definition}

The main lemma we are going to prove about controlled directed bipartite graphs is the following:

\begin{lemma}\label{lem:rcdbg}
  There exists a function $F \st \N^3 \rightarrow \N $ such that if
  $(G := (A \dot\cup B, E), \beta, \lambda, \eta)$ is an
  $r$-controlled directed bipartite graph with $|B|\geq F(r,p, q)$
  then
  \begin{itemize} \squishlist
  \item $\exists S \in [A]^{\leq \binom{q}{2}}$ such that $G - S$
    contains a $1$-scattered set of size $p$; or
  \item $G$ contains a controlled crown of order $q$.
  \end{itemize}
\end{lemma}

To see what this is good for, we first go on proving
Theorem~\ref{thm:main_equivalence} assuming Lemma~\ref{lem:rcdbg} and
present the proof of Lemma~\ref{lem:rcdbg} afterwards. Both of these
proofs are based on several other intermediate lemmata.

\begin{lemma}\label{lem:main-tec}
  Let $G$ be a directed graph, $r \geq 0$, and $p, q>0$. Let $I$ be an
  $r$-scattered set in $G$ of size at least $F(r,p,q)$ where $F$ is the
  function defined in Lemma~\ref{lem:rcdbg}.  Then $S_q\dminor_{r} G$ or
  there is a set $S \in [V(G)]^{\leq \binom{q}{2}}$ s.t.\ $G - S$
  contains an $(r+1)$-scattered set of size $p$.
\end{lemma}
\begin{proof}
  For $u\in I$ let 
  \[
    P(u) := \{ v \st \text{there is a path of length $\leq r$ from
      $v$ to $u$ } \}.
  \]
  By construction, $P(v) \cap P(u) = \emptyset$ whenever $u\not=v\in
  I$.

  We construct the following controlled directed bipartite graph $(H
  := (A\dot\cup B, E), \beta, \lambda, \eta)$. Set $B := I$. For each
  vertex $v\in V(G)$, if there are at least $2$ distinct $u, u'\in I$
  reachable from $v$ in $G$ in at most $r+1$ steps, then we add
  $\tilde v$ to $A$ and initialise $\lambda(\tilde v) := r+1$. Then we
  do the following: for each pair $\tilde v \in A$ and $u \in B$ such
  that $v$ can reach $u$ in $G$ in at most $r+1$ steps, we fix a
  shortest path $v=v_i \dots v_0=u$ of length $i \leq r+1$ in $G$. We
  add an edge $\tilde vu$ to $E$, label $e$ by $\eta(e) :=
  \set{v_{i-1},\dots,v_0}$, and update $\lambda(\tilde v) := \min
  \set{i, \lambda(\tilde v)}$. Note that in particular, we include
  edges $e = \tilde uu$ for all $u \in I$ and label it by $\eta(e) =
  \emptyset$.  If $v\in P(u)$ for some $u\in I$, then we define
  $\beta(\tilde v) := \set{u}$ and otherwise $\beta(\tilde v) :=
  \emptyset$. Note that by construction, (i) $v$ can be in at most one set
  $P(u)$; (ii) $\beta(\tilde v) = \emptyset \Leftrightarrow
  \lambda(\tilde v) = r+1$; and (iii) our construction of $\beta$ and
  $\eta$ is fully conform to Definition~\ref{def:rcdbg}~(a).

  By Lemma~\ref{lem:rcdbg}, either there is a set $\tilde S\subseteq A$ of
  size at most $({q\atop 2})$ and a set $I'\subseteq B$ of size $p$
  which is $1$-scattered in $H - S$, or there is a controlled crown
  $S_q := (A' \cup B', E')\subseteq H$ of order $q$ contained in $H$.

  In the first case, let $S \subseteq V(G)$ be the vertices in $G$
  corresponding to $\tilde S$. We claim that $I'\subseteq I$ is
  $(r+1)$-scattered in $G - S$. For, suppose there were $v, u, u'$
  with $v\in V(G)\setminus S$ and $u,u'\in I'$ such that both $u, u'$
  are reachable from $v$ in at most $r+1$ steps. Then $\tilde v\in A$
  and there are edges from $\tilde v$ to $u$ and $u'$ in $H$,
  contradicting the fact that $I'$ was $1$-scattered in $H - S$.

  So suppose $S_q := (A' \cup B', E') \subseteq H$.  We claim that
  $S_q \dminor_r G$ as follows. For $\tilde a\in A'$ let $T_a$ be the
  tree consisting only of $a$. For $b\in B'$ let $T_b$ be a spanning
  in-branching of $\bigcup_{e=ab \in E'} \eta(e)$, and note that $T_b
  \subseteq P(b)$ clearly exists. By Definition~\ref{def:rcdbg}~(b),
  $A' \cap \eta(e) = \emptyset$ whenever $e \in E'$ and therefore $T_a
  \cap T_b = \emptyset$ for each $\tilde a\in A'$ and $b\in B'$. Also,
  if $\tilde ab\in E'$ then $a$ has an edge to $\eta(e)$ and hence
  $T_a$ has an edge to $T_b$. This shows that $\set{T_a, T_b \st
    \tilde a\in A', b\in B'}$ are the branch sets of an $S_q$-minor of
  $G$ of depth $r$.\qed
\end{proof}

We would like to emphasise that Lemma~\ref{lem:main-tec} works in
particular for $r=0$, i.e.\ when $I$ is an arbitrary set in $G$ which
is large enough. Now we obtain our main theorem by repeated
application of the above.

\begin{lemma}
  If $\CCC$ is \nwcrown then it is uniformly quasi-wide.
\end{lemma}
\begin{proof}
  Let $f \st \N \cup \set{0} \rightarrow \N$ be the function such that
  $S_{f(r)} \not\dminor_r G$ for all $G\in \CCC$. We construct the
  margins $s(r)$ and $N(r, m)$ witnessing that $\CCC$ is uniformly
  quasi-wide. Let $F$ be the function from Lemma~\ref{lem:rcdbg}. We
  define $\tilde N(r,m,r) := m$, and $\tilde N(r,m,i) :=
  F(r,\tilde N(r,m,i+1), f(i))$. We let $N(r,m) := \tilde N(r,m,0)$ and
  $s(r) := \sum_{i=0}^{r-1}\binom{f(i)}{2}$. Let $W$ be a set of size
  at least $N(r,m)$ in $G \in \CCC$.

  We let $I_0 := W$, $Z_0 := \emptyset$. Assuming $I_i$ and $Z_i$ are
  given and that $I_i$ is $i$-scattered in $G-Z_i$ with $|I_i| \geq
  \tilde N(r,m,i)$, we construct $I_{i+1}$ and $Z_{i+1}$ as
  follows. By Lemma~\ref{lem:main-tec}, we know that either $S_{f(i)}
  \dminor_i G-Z_i \dminor_i G$ or there exists a set $Z'_i$ of size at
  most $\binom{f(i)}{2}$ and $I_{i+1} \subseteq I_i$ of size $\tilde
  N(r,m,i+1)$ such that $I_{i+1}$ is $(i+1)$-scatterd in
  $G-Z_i-Z'_i$. Since the former cannot happen in our \nwcrown class,
  the latter must be the case, and our claim is proved by defining
  $Z_{i+1} := Z_i \cup Z'_i$ and letting $S := Z_r$ and $U := I_r$
  be the output of the algorithm.\qed
%
\end{proof}

This finishes the proof of Theorem~\ref{thm:main_equivalence} based on
Lemma~\ref{lem:rcdbg}. We proceed by proving the latter in a bottom-up
fashion. Let us first state the following simple observation:

\begin{lemma}\label{lem:cc2}
  If $S_q=(A' \dot\cup B', E')$ is a crown in an $r$-controlled
  directed bipartite graph $\rcdbg$ such that for all $a \in A'$,
  $\beta(a) \cap B' = \emptyset$, then $S_q$ is a controlled crown.
\end{lemma}
\begin{proof}
  Suppose there is a $v \in A'$ and $e=ab \in E'$ such that $v \in
  \eta(e)$. But then we know by definition that $\beta(v) = b \in B'$,
  a contradiction.\qed
\end{proof}

Our plan is to apply several Ramsey-type arguments -- which we state
from an algorithmic viewpoint -- to either find a large $1$-scattered
set, a controlled crown by way of Lemma~\ref{lem:cc2}, or a controlled
crown in which \emph{all vertices of $A$ have the same level according
  to $\lambda$}. The latter restriction is crucial to our proof and
will be applied when we invoke Lemma~\ref{lem:clique}, a very useful
auxiliary lemma stated below.

\begin{lemma}\label{lem:clique}
  There is a computable function $f \st \N\rightarrow \N$ such that if
  $G := K_{f(n)}$ and $\gamma\st E(G) \rightarrow [V(G)]^{\leq 1}$
  s.t. $\gamma(e) \cap e = \emptyset$ for all $e\in E(G)$ then
  there is $H \isom K_n \subseteq G$ such that $\gamma(e) \cap V(H) =
  \emptyset$ for all $e\in E(H)$.
\end{lemma}
\begin{proof}
  Let $R(n)$ denote the $n^{\mbox{\smaller th}}$ Ramsey number, i.e.\ the least integer
  $m$ such that if we $2$-color the edges of $K_m$, there always
  exists a monochromatic $K_n$.  We define the function $f(n)$
  recursively as $f(1) := 1$ and $f(n+1) := 1+R(2\cdot f(n))$.

  The lemma is proved by induction on $n$. For $n=1$ there is nothing
  to prove as any $K_1$ contains a $K_1$ as required.

  Suppose the statement is proved for $n$ and let $G \isom K_{f(n+1)}$
  and $\gamma$ be as required. Choose a vertex $v\in V(G)$ and colour all
  edges $e$ in $G_v := G - v$ by $v$ if $\gamma(e) = \{v\}$ and by $\bar v$
  otherwise. By Ramsey's theorem, as $|G_v| = R(2\cdot f(n))$ either
  there is a set $X\subseteq V(G_v)$ of size $2\cdot f(n)$ such that
  all edges between elements of $X$ are coloured $v$ or there is such
  a set where all edges are coloured $\bar v$. In the first case, any
  $X'\in [X]^{n+1}$ induces the required graph $H\isom K_{n+1}\subseteq G$.

  So suppose $X$ induces a sub-graph where all edges are labelled
  $\bar v$. 
  We construct a set $X'$ as follows. Initially, set $X' :=
  \emptyset$. While $X\not=\emptyset$, choose a vertex $u\in X$. Add
  $u$ to $X'$ and remove $u$ and $\gamma(\{v,u\})$ from $X$. As we are
  removing at most $2$ elements of $X$ in each step, we get a set $X'$
  of size at least $f(n)$ with the property that for each $u\in X'$,
  $\gamma(\{v,u\}) \cap X' = \emptyset$. 

  Let $G' := G[X']$ and $\gamma'(e) := \gamma(e)$ for all $e\in E[G']$. As $|G'|
  \geq f(n)$, by the induction hypothesis, $G'$ contains a subgraph
  $H'\isom K_{n}$ such that $\gamma(e) \cap V(H') = \emptyset$ for all
  $e\in E(H')$. Hence, $H := G[V(H')\cup \{ v\}]$ is the required
  subgraph of $G$ isomorphic to $K_{n+1}$ with $\gamma(e)\cap V(H) =
  \emptyset$ for all $e\in E(H)$.\qed
\end{proof}

Next, we state and prove Lemma~\ref{lem:0}. An interesting point about
this lemma is that it somewhat surprisingly includes two unrelated
applications of Lemma~\ref{lem:clique}, and in a sense, bears the
delicate part that makes our overall proof work correctly.

\begin{lemma}\label{lem:0}
  Let $(G:=(A\dot\cup B,E),\beta,\lambda,\eta)$ be an $r$-controlled
  directed bipartite graph such that  
  \begin{itemize} \squishlist
    \item[(i)] $\lambda$ is a constant function equal to $c \in \set{0,\dots,r+1}$;
    \item[(ii)] any two vertices in $B$ have a common neighbour in $A$; and
    \item[(iii)] every vertex in $A$ has at most $n$ successors. 
  \end{itemize}
  If $|B| \geq g(q,n):=(2n)^{2f(q)}$, where $f$ is as in
  Lemma~\ref{lem:clique}, then $G$ contains a controlled crown of
  order $q$.
\end{lemma}
\begin{proof}
  For $u,v \in B$, let us say that $u$ and $v$ have a \emph{red
    connection} if there is an $a \in A$ with $u,v \in N^+(a)$ such
  that $\beta(a) \neq u$ and $\beta(a) \neq v$. If $u$ and $v$ are
  both neighbors to a vertex $a' \in A$ with $\beta(a') = u$ or
  $\beta(a') = v$, we say $u$ and $v$ have a \emph{yellow
    connection}. Note that $u$ and $v$ can have both red and yellow
  connections. If the vertices $a,a'$ do not have any other
  neighbours, we say the connection is \emph{pure}.

  Let $D_0 := B$ and let $R_0 := Y_0 := \emptyset$. Suppose $D_i$,
  $R_i$, and $Y_i$ have been defined maintaining the invariants that
  \begin{itemize} \squishlist
    \item any vertex $v$ in $R_i$ has a pure red connection to any vertex
      in $R_i \cup D_i \setminus \set{v}$; and
    \item any vertex $v$ in $Y_i$ has a pure yellow connection to any
      vertex in $Y_i \cup D_i \setminus \set{v}$.
  \end{itemize}
  Choose an arbitrary vertex $v_{i+1}\in D_i$, remove it from $D_i$,
  and set $D_R=D_Y=\emptyset$. For $j = 1\dots 2 \cdot
  (2n)^{2f(q)-(i+1)}$ do the following.  Choose $u_j\in D_i$ and $a\in
  A$ with $\set{v_{i+1}, u_j} \subseteq N^+(a)$. If $a$ induces a red
  connection between $v_{i+1}$ and $u_j$, put $u_j$ in $D_R$,
  otherwise put $u_j$ in $D_Y$. Delete $N^+(a)$ from $D_i$.

  As no $a\in A$ has more than $n$ successors, this process can
  complete successfully. At this time, either $D_R$ or $D_Y$ have at
  least $(2n)^{2f(q)-(i+1)}$ elements. If $|D_R| \geq |D_Y|$, define
  $D_{i+1} := D_R$, $R_{i+1} := R_i \cup \set{v_{i+1}}$, and $Y_{i+1}
  := Y_i$; otherwise let $D_{i+1} := D_Y$, $Y_{i+1} := Y_i \cup
  \set{v_{i+1}}$, and $R_{i+1} := R_i$. After $2f(q)$ iterations,
  either $R_{2f(q)}$ or $Y_{2f(q)}$ contain at least $f(q)$ vertices.

  Case~(i), $|R_{2f(q)}| \geq f(q)$: We construct a complete
  undirected graph $G_R := (V_R,E_R)$ with $V_R := R_{2f(q)}$. For
  vertices $u,v \in V_R$ and $e=uv \in E_R$, let $a \in A$ be the
  vertex inducing the pure red connection between $u$ and $v$, and
  define $\gamma(e) := \beta(a)$. As $\beta(a) \neq u$ and $\beta(a)
  \neq v$, we have that $G$ and $\gamma$ fulfil the requirements of
  Lemma~\ref{lem:clique}, and we obtain a clique $H_R \subseteq G_R$ of
  size $q$ such that $\gamma(e) \cap V(H_R) = \emptyset$ for all $e
  \in E(H_R)$. This translates to a controlled crown of order $q$ in
  $G$ by way of Lemma~\ref{lem:cc2}.

\begin{figure}[t]
  \centering
  \includegraphics[height=2.5cm]{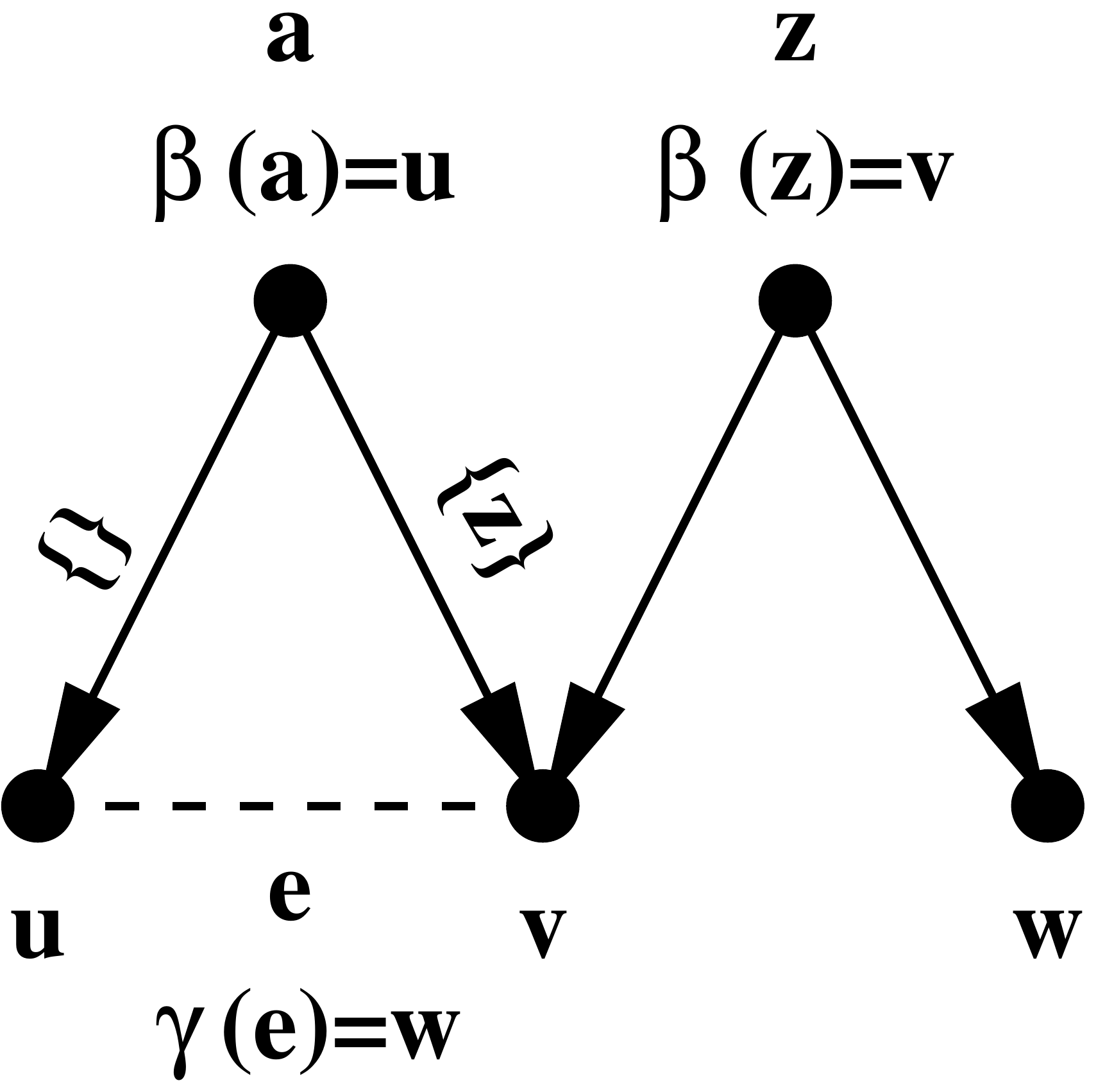}
  \caption{Illustration of how to define $\gamma$ in case~(ii) of the
    proof of Lemma~\ref{lem:0}.}
  \label{fig:rcdbg-proof} 
\end{figure}

Case~(ii), $|Y_{2f(q)}| \geq f(q)$: We construct a complete undirected
graph $G_Y := (V_Y,E_Y)$ with $V_Y := Y_{2f(q)}$. As shown in
Figure~\ref{fig:rcdbg-proof}, for vertices $u,v \in V_Y$ and $e=uv \in
E_Y$, let $a \in A$ be the vertex inducing the pure yellow connection
between $u$ and $v$, and w.l.o.g.\ assume $\beta(a) = u$. Let $e_1 =
au$ and $e_2 = av$. Note that for all $z' \in \eta(e_1)$, we know that
$\lambda(z') < \lambda(a)=c$, and hence $z' \notin A$. In $\eta(e_2)$,
there exists at most one $z$ with $\lambda(z)=c$. If no such $z$
exists or $z$ does not occur as a pure yellow connection between
vertices of $Y_{2f(q)}$, define $\gamma(e) := \emptyset$. Otherwise,
we know by Definition~\ref{def:rcdbg}~(a) that $\beta(z)=v$, and hence
$z$ induces a pure yellow connection between $v$ and some vertex $w
\in Y_{2f(q)}$. If $w=u$, then $z$ is redundant and we delete it from
$G$ and set $\gamma(e) = \emptyset$. Otherwise, we define $\gamma(e) =
w$. Now, if we apply Lemma~\ref{lem:clique} to obtain a clique $H_Y
\subseteq G_Y$ of size $q$ such that $\gamma(e) \cap V(H_Y) =
\emptyset$ for all $e \in E(H_Y)$, it means that for all $e \in
E(H_Y)$ such a vertex $w$ and therefore its pure connection via $z$ do
not occur in the crown induced by $V(H_Y)$ in $G$, i.e.\ $V(H_Y)$
induces a controlled crown in $G$.\qed
\end{proof}

\begin{lemma}\label{lem:1}
  Let $f \st \N^4 \rightarrow \N$ be defined as $f(r, p, q, n) :=
  (r+3)^{p+(r+2)\cdot g(q,n)}$, where $g(q,n)$ is the function defined in
  Lemma~\ref{lem:0}. Let $(G := (A,\dot \cup B, E), \beta, \lambda,
  \eta)$ be an $r$-controlled directed bipartite graph.  For all $r
  \geq 0$ and $p, q, n>0$, if $|B| > f(r, p, q, n)$ then
  \begin{enumerate}\squishlist
  \item $A$ contains a vertex with $n+1$ successors; or
  \item $B$ contains a $1$-scattered set of size $p$; or
  \item $G$ contains a controlled crown of order $q$.
  \end{enumerate}
\end{lemma}
\begin{proof}
  If there is a vertex $a\in A$ of degree at least $n+1$ we are
  done. So suppose all vertices $a\in A$ have degree at most $n$.  We
  are going to define a sequence of sets $B_i, I_i,
  D^0_i,\dots,D^{r+1}_i$ with the invariant that
  \begin{itemize} \squishlist
  \item  $|B_i| \geq (r+3)^{p+(r+2)\cdot g(q,n)-i}$,
  \item no two $v \neq v'\in
    I_i$ and no $v\in I$ and $u\in B_i$ have a common predecessor in $A$,
  \item for $j=0,\dots,r+1$, any two $v, v'\in D^j_i$ and any $v\in
    D^j_i$ and $u\in B_i$ have a common predecessor in $a \in A$ with
    $\lambda(a)=j$.
  \end{itemize}

  Set $I_0 = D^0_0 = \dots = D^{r+1}_0 := \emptyset$ and $B_0 := B$
  which clearly satisfies the invariant. Now suppose $i<p+(r+2)\cdot g(q,n)$ and
  $B_i, I_i, D^0_i,\dots,D^{r+1}_i$ have been defined. Choose $v\in B_i$. For $j=0,\dots,r+1$, define $A^j := \set{a \in N^-(v) \st \lambda(a) = j}$ and $B^j := N^+(A_j) \cap B_i$. Let $B' := B_i \setminus \bigcup_j B^j$.
  \begin{itemize}
  \item If $|B'| \geq \frac{1}{r+3} |B_i|$ then $B_{i+1} := B'$,
    $I_{i+1} := I_i \cup \set{v}$, and $D^j_{i+1} := D_i$. Clearly,
    this maintains the invariant.
  \item Otherwise, there must exist a $t$ s.t.\ $|B^t| \geq
    \frac{1}{r+3} |B_i|$; define $B_{i+1} := B^t$, $I_{i+1} := I_i$,
    $D^t_{i+1} := D^t_i \cup \{ v\}$, and $D^j_{i+1} := D^j_i$ for all
    $j \neq t$. Again, the invariant is maintained.
  \end{itemize}
  
  Now if for some $i$ we have $|I_i| = p$, then we have found a
  $1$-scattered set of size $p$. Otherwise, there must exist a $t$
  s.t.\ $|D^t_{p+(r+2)\cdot g(q,n)}| \geq g(q,n)$. Applying
  Lemma~\ref{lem:0} to this set results in a controlled crown of order
  $q$.\qed
\end{proof}

Finally, we are ready to prove Lemma~\ref{lem:rcdbg}, concluding this
section.

\begin{proof}[of Lemma~\ref{lem:rcdbg}]
  Let $\tilde F(r, p, q, 0) := q$ and $\tilde F(r, p, q, t) := f(r, p,
  q, \tilde F(r, p, q, t-1))$, where $f$ is as in
  Lemma~\ref{lem:1}. Define $F(r, p, q) := \tilde
  F(r,p,q,\binom{q}{2})$ and assume $|B| \geq F(r, p,q)$.

  We are going to construct a sequence $S_i \subseteq A, B_i\subseteq
  B$ as follows, where we will maintain the invariant that $|B_i| \geq
  \tilde F(r,p,q,\binom{q}{2}-i)$ and that for all $v\in S_i$,
  $B_i\subseteq N^+(v)$ and $\beta(v) \cap B_i = \emptyset$.

  Set $S_0 := \emptyset$ and $B_0 := B$. Now suppose $S_i, B_i$ have
  been defined.  If $A\setminus S_i$ contains a vertex $a$ such that
  $|N^+(v) \cap B_i| > \tilde F(r,p, q, \binom{q}{2}-i-1)$ then set
  $S_{i+1} := \set{v} \cup S_i$ and $B_{i+1} := (B_i \cap N^+(v))
  \setminus \beta(v)$. Clearly, this maintains the invariant.

  If we can continue this process until we obtain a set $S :=
  S_{\binom{q}{2}}$, then $S$ together with any subset of
  $B_{\binom{q}{2}}$ of size $q$ contain an $H\isom S_q$, which is in
  fact a controlled crown by Lemma~\ref{lem:cc2}.

  Otherwise, if at some point $A\setminus S_i$ does not contain a
  vertex $a$ such that $|N^+(v) \cap B_i| > \tilde F(r, p, q,
  \binom{q}{2}-i-1)$, then, applying Lemma~\ref{lem:1} to $(G' :=
  G[A\cup B_i \setminus S_i], \beta',\lambda,\eta)$, where $\beta'(v)
  := \emptyset$ if $\beta(v) \not\in B_i$ and $\beta'(v) := \beta(v)$
  otherwise, we obtain either a controlled crown of order $q$ or a
  $1$-scattered set $I \subseteq B_i$ of size $p$ in $G'$. But then,
  setting $S := S_i$ we see that $I$ is $1$-scattered in $G - S$.\qed
\end{proof}

\section{Algorithmic Results}
\label{sec:algo}

The aim of this section is to show that problems such as
directed dominating sets, directed independent dominating set (Kernel problem),
and many variations are fixed-parameter tractable on any class of
directed graphs which is \nwcrown. Domination problems on
nowhere dense classes of undirected graphs have been studied
in \cite{DawarKre09} and some of the ideas developed there can be
adapted to the directed setting.

As a first step towards algorithmic applications we show that
we can efficiently compute scattered sets in graph classes which are
effectively uniformly quasi-wide.

\begin{lemma}\label{lem:comp-S}
 Let $\CCC$ be effectively uniformly quasi-wide with margin $s, N$.
 The following problem is fixed-parameter tractable.
 \pprob{8cm}
 {$G\in\CCC$, $d, m\in\N$, $W\subseteq \vrtG$ such that $|W|\geq N(d,
   m)$}%
 {$d+m$}%
 {Compute $S\subseteq \vrtG$, $|S|\leq s(d)$ and $U\subseteq W$,
   $|U| = m$ such that $U$ is $d$-scattered in $G - S$}
\end{lemma}
\begin{proof}
  To compute $U$, we first choose an arbitrary subset $U'\subseteq W$
  of size exactly $N(d, m)$.  By definition of uniformly
  quasi-wideness, we are guaranteed to find a set $S$ and a set
  $U\subseteq U'$ as required. But $|U'|$ depends only on the
  parameter and hence we can simply traverse through all subsets $U$
  of $U'$ of size at least $m$. For each such set $U$ we compute $C :=
  \bigcup_{u\not=u' \in U} N^-_d(u) \cap N^-_d(u')$. Clearly,
  $U\setminus C$ is $d$-scattered in $G - C$ and $C$ can be computed
  in time $\Oof(|U|^2\cdot |G|)$.  Hence, if $|U\setminus C| \geq m$
  and $|C|\leq s(d)$ we return $U$ and $S := C$. \qed
\end{proof}

We briefly recall the definition of the algorithmic problems we refer
to below. As defined in Section~\ref{sec:prelim}, a
\emph{$d$-dominating set} is a set $D \subseteq V(G)$ such that
$N_d^+(D) = V(G)$ and is just called a \emph{dominating set} for
$d=1$.  An \emph{independent dominating set} is a dominating set $D$
which is itself independent, i.e.~$(u,v)\not\in E(G)$ for all $u,v\in
D$. The corresponding decision problems are defined as the problem,
given a digraph $G$ and a number $k$, to decide whether $G$ contains
an (independent or unrestricted) $d$-dominating set of size $k$. The
parameter is $k+d$.

Recall that an important variant of the undirected dominating set
problem is the \emph{connected dominating set problem}, where we are
asked to find a dominating set $D$ of size $k$ such that $D$ induces a
connected subgraph. There are various natural translations of this
problem to the directed case: we can require the dominating set to
induce a strongly connected subgraph or we can simply require it to
induce an out-branching, i.e.~a directed tree whose edges are oriented
away from a root. The second variation, which we call \emph{dominating
  out-branching}, still captures the idea that information can flow
from the root to all vertices in the dominating set.

We now establish our main algorithmic result of the paper where we
show that many covering problems become fixed-parameter tractable on
classes of directed graphs which are \nwcrown.

\begin{theorem}\label{theo:main-alg}
  Let $\CCC$ be a class of directed graphs which is {\nwcrown}. Then
  the directed (independent or unrestricted) dominating set problem,
  the dominating out-branching problem, and the independent set problem
  as well as their distance-$d$-versions are fixed-parameter tractable
  on $\CCC$.
\end{theorem}

We present the proof of Theorem~\ref{theo:main-alg} for the cases of
independent dominating set, $d$-dominating set, and dominating
out-branching. This captures the essence of the ideas involved in
designing such algorithms for directed uniformly quasi-wide
classes. The other problems and some variations such as Red-Blue
domination or Roman domination follow similarly. See \cite{Cesati06}
for detailed definitions of these problems.

\begin{proof}[of Theorem~\ref{theo:main-alg} for the case of independent dominating set]
  Let $G \in \CCC$ and $k$ be given. We solve a slightly more general
  version, where we are additionally given a set $Y \subseteq \vrtG$
  and require that the solution belongs to $\vrtG \setminus Y$.  If
  $|\vrtG| \leq N(1, k+1)$, then we solve the problem by exhaustively
  searching for all possible sets of $k$ vertices.
 
  Otherwise, we apply Lemma~\ref{lem:comp-S} to obtain a set
  $U\subseteq \vrtG$ of size $k+1$ and a set $S$ of at most $s(1)$
  vertices such that $U$ is $1$-scattered in $G - S$.  As $U$ is
  $1$-scattered and of size $k+1$, it follows that any dominating set
  must contain at least one vertex from $S$. Hence, it suffices that
  for each $v \in S \setminus Y$, we recursively solve the problem on
  $G' := G - s - N^+(s)$ and $Y' := Y \cup N^-(s)$ with $k' := k-1$
  and report a positive answer in case one is found. This results in a
  recursion tree of depth at most $k$ and degree at most $s(1)$ which
  are both bounded in terms of our parameters.\qed
\end{proof}

Next, we consider the $d$-dominating set problem. The following
lemmata are proved for the undirected version in~\cite{DawarKre09} and
the proofs work exactly as they are for the directed version as well.

\begin{lemma}\label{lem:d-dom-intermediate}
  The following problem is fixed-parameter tractable.
  
  \pprob{11cm}%
  {a directed graph $G$, $W \subseteq V(G)$ and $k,d \in \N$}%
  {$|W|, k, d$}%
  {Does there exist a set $X\in [V(G)]^{\leq k}$ such that $X$ $d$-dominates $W$?}
\end{lemma}

\begin{lemma}\label{lem:d-dom-reduce}
  Let $\CCC$ be effectively uniformly quasi-wide with margin $s, N$.
  The following problem is fixed-parameter tractable. Furthermore,
  such a vertex $w$ always exists.
  \pprob{11cm}
  {$G\in\CCC$, $k, d\in\N$, $W\subseteq \vrtG$ such that $|W| > N(d,
    (k+2)(d+1)^{s(d)})$}%
  {$k+d$}%
  {Compute a vertex $w \in W$ such that for any set $X \in
    [V(G)]^{\leq k}$, $X$ $d$-dominates $W$ if and only if $X$
    $d$-dominates $W \setminus \set{w}$.}
\end{lemma}

In order to solve the $d$-dominating set problem, we start with $W =
V(G)$ and apply Lemma~\ref{lem:d-dom-reduce} until the size of $W$
becomes bounded by $N(d,(k+2)(d+1)^{s(d)}$. Then we can apply
Lemma~\ref{lem:d-dom-intermediate}.

Finally, let us consider the dominating out-branching problem. We first need the following lemma, which is proved by using the fixed-parameter tractability of \emph{directed Steiner tree} in general graphs.

\begin{lemma}\label{lem:dom-outb-intermediate}
  The following problem is fixed-parameter tractable.
  
  \pprob{11cm}%
  {a directed graph $G$, $u_1, \dots, u_t \in V(G)$ and
    $W\subseteq V(G)$, $j\in\N$}%
  {$|W|, t+j$}%
  {Do there exist distinct vertices $v_1, \dots, v_j$ such that $G[v_1, \dots,
    v_j, u_1, \dots, u_t]$ is an out-branching and $W \subseteq
    N^+(v_1, \dots, v_j)$?}
\end{lemma}
\begin{proof}
  Suppose there are distinct vertices $v_1, \dots, v_j$ such that
  $G[v_1, \dots, v_j, u_1, \dots, u_t]$ is an out-branching and $W
  \subseteq N^+(v_1, \dots, v_j)$. Then we can partition $W$ into
  (possibly empty) sets $W_1, \dots, W_j$ so that $W_i \subseteq
  N^+(v_i)$. Let
  \[
    X_i := \{ x \in V(G) \st W_i \subseteq N^+(x) \} \, ;
  \]
  then $v_i \in X_i$ for all $i$. Hence, to decide whether such
  vertices exists it suffices to compute all possible partitions of
  $W$ into sets $W_1, \dots, W_j$ (some of which may be empty), for
  each partition the sets $X_i$, and then check if an out-branching of
  size at most $t+j$ exists that includes $u_1,\dots,u_t$ and at least
  one element out of each $X_i$.

  But this problem can easily be transformed to a \emph{directed
    Steiner tree} instance: for each $X_i$ add a veretex $x_i$ and for
  each $v \in X_i$ a directed edge $vx_i$. Now, the goal is to find a
  directed Steiner tree, i.e.\ an out-branching, of size at most
  $t+2j$ on the terminal set $\set{u_1,\dots,u_t,x_1,\dots,x_j}$. Since the
  directed Steiner tree problem is fixed-parameter tractable (see
  e.g.\ Guo et al.~\cite{GuoNS09}) and the size of $W$ is part of the
  parameter, all steps can be erformed in fpt time. \qed

%
%
%
\end{proof}

\begin{proof}[of Theorem~\ref{theo:main-alg} for the case of
  dominating out-branching]
  Let $G \in \CCC$ and $k$ be given.  We will recursively solve a
  slightly more general problem where the input consists of a graph
  $G$, a set $W\subseteq V(G)$, a number $j$ and vertices $u_1, \dots,
  u_t \in V(G)$.  The problem is to decide whether there are $j$
  vertices $v_1, \dots, v_j$ such that $G[u_1, \dots, u_t, v_1, \dots,
    v_j]$ is an out-branching and $\{v_1, \dots, v_j\}$ dominate
  $W$. The parameter is $t+j$.

  Solving the problem for $W := V(G)$, $t:= 0$ and $j := k$ solves the
  dominating out-branching problem.

  If $|W| \leq N(1, t+j+1)$, then we can apply
  Lemma~\ref{lem:dom-outb-intermediate}.  Otherwise, let $S\subseteq
  V(G)$ and $A\subseteq W$ be such that $|A|>t+j$ and $A$ is
  $1$-scattered in $G - S$. This implies that every dominating set
  must contain a vertex of $S$. For each $u_{t+1}\in S$ we call the
  algorithm recursively on $G, u_1, \dots, u_{t+1}, W\setminus
  N^+(u_{t+1}), j-1$ to check whether $u_1, \dots, u_{t+1}$ can be
  extended to a dominating out-branching.  This results in a recursion
  tree of depth at most $k$ and degree at most $s(1)$ which are both
  bounded in terms of our parameters.\qed
\end{proof}

To adapt this proof idea to a solution of the strongly connected
dominating set problem we would need to solve the \emph{directed
  strongly connected Steiner subgraph} problem, i.e.~the problem,
given vertices $t_1, \dots, t_p$ in a digraph $G$ and a number $k$, to
decide whether there are $k$ vertices $v_1, \dots, v_k$ so that $t_1,
\dots, t_p, v_1, \dots, v_k$ induce a strongly connected
sub-graph. This problem is known \emph{not} to be fixed-parameter
tractable (with parameter $k+p$) on general directed graphs but it
might become fpt on {\nwcrown} classes of graphs. We leave this for
future research.


%





\section{Conclusion}
\label{sec:conclusion}

In this paper we have given a classification of classes of directed
graphs that is fundamentally different from the existing proposals for
directed width measures based on tree-width. We have seen that even
for the very general concept of {\nwcrown} classes of directed graphs,
``covering'' problems such as the (independent) dominating set problem become
fixed-parameter tractable, problems that are not tractable on classes
of bounded width for any other width measure (except bi-rank width,
see the introduction). 

We believe that our proposal here may lead to a new and promising
structural theory for directed graphs with algorithmic applications in
mind. However, this clearly is only a first step and many questions
remain open. 

In particular, linkage problems become (XP-)tractable on classes of
directed graphs of bounded directed tree-width
\cite{JohnsonRobSeyTho01} but we do not presently know whether they
are tractable on any of our classes. But from this it follows that if
we consider classes of directed graphs which have bounded directed
tree-width and at the same time are {\nwcrown}, then on such classes
we can solve a wide range of problems efficiently, both linkage and
covering problems. Clearly, this definition is extremely technical and
as such not very useful. But it shows that there are structural
concepts for directed graphs which do allow for a very rich
algorithmic theory and we believe it would be worth studying these
concepts in more detail.

Additionally, from a structural point of view, the perspective
presented in our work opens a number of intriguing problems to look
at. For example, a slight adaption of Theorem~\ref{thm:alternating}
shows that if a class of digraphs has unbounded undirected tree-width,
then it contains arbitrarily large alternating cycles, i.e.\ an
orientation of a cycle that contains a long alternating path minor. In
terms of the directed minor relation, this results in an
\emph{anti-chain} in the class. One question is, if we exclude such
anti-chains, does a minor-closed class become \emph{well
  quasi-ordered}? Since such a class must have bounded tree-width, it
seems plausible that the answer might be positive.

Another structural question relates to \emph{colourings} of our
classes. It is easily seen that classes of directed bounded expansion
can be coloured by a constant number of colours. However, could it be
possible that more powerful colourings are possible, such as acyclic
colourings or colourings that induce some other bounded width measure
(cf.\cite{NesetrilOss08})? This seems to be a quite difficult question as known
methods for undirected graphs do not generalise. However, it might be
that at least on the intersection of classes of directed bounded
expansion and crown-minor-free, one can obtain some results.

\bibliographystyle{alpha}
\bibliography{papers}

\end{document}